%% file: neurips_2026.tex
\documentclass{article}

\usepackage[main, final]{neurips_2026}

\usepackage[utf8]{inputenc}
\usepackage[T1]{fontenc}
\usepackage{hyperref}
\usepackage{url}
\usepackage{booktabs}
\usepackage{amsfonts}
\usepackage{nicefrac}
\usepackage{microtype}
\usepackage{xcolor}

\usepackage{amssymb}
\usepackage{amsmath}
\usepackage{amsthm}
\usepackage{graphicx}
\usepackage{multirow}
\usepackage{array}
\usepackage{adjustbox}
\usepackage{pifont}
\usepackage{enumitem}
\usepackage{algorithm}
\usepackage{algorithmic}
\usepackage{wrapfig}

\newtheorem{theorem}{Theorem}
\newtheorem{lemma}{Lemma}
\newtheorem{assumption}{Assumption}
\newtheorem{definition}{Definition}
\theoremstyle{remark}
\newtheorem*{remark}{Remark}

\title{Aligning Shared and Routed Experts for Cross-Subject EEG Generalization}

\author{
  Zhi Zhang$^{1}$ \quad
  Yan Liu$^{1}$\thanks{Corresponding author.} \quad
  Zhejing Hu$^{1}$ \quad
  Gong Chen$^{2}$ \quad
  Sheng-hua Zhong$^{3}$ \\
  \bfseries Changhong Jing$^{1,4}$ \quad
  Shuqiang Wang$^{4}$ \quad
  Jibin Wu$^{1}$ \quad
  KC Tan$^{1}$ \quad
  Jiannong Cao$^{1}$ \\[6pt]
  \mdseries $^{1}$The Hong Kong Polytechnic University, Hong Kong, China \\
  $^{2}$Biaohan Limited, Shenzhen, China \quad
  $^{3}$Shenzhen University, Shenzhen, China \\
  $^{4}$Shenzhen Institutes of Advanced Technology, Chinese Academy of Sciences, Shenzhen, China \\[3pt]
  \texttt{\{zhi271.zhang, zhejing.hu, changhong.jing\}@connect.polyu.hk} \\
  \texttt{\{yan.liu, jibin.wu, kaychen.tan, jiannong.cao\}@polyu.edu.hk} \\
  \texttt{heinz.g.chen@hotmail.com, csshzhong@szu.edu.cn, sq.wang@siat.ac.cn}
}

\begin{document}

\maketitle

\begin{abstract}
Cross-subject EEG generalization is challenging due to substantial heterogeneity across subjects. Existing methods typically learn either a shared subject-invariant model or multiple subject-specialized experts, but these two paradigms fail in complementary ways: the former may over-reduce subject-specific discriminative signals, while the latter may under-reduce transferable structure. We show that their suitability depends on the reducibility cost of branch-specific functions to branch-invariant ones, and we further provide a theory-to-method mapping that instantiates alignment principles in cross-subject EEG learning. Based on this insight, we propose Shared-Routed Expert Alignment (SREA), a collaborative framework that couples a shared expert for reducible invariant functions with routed experts for irreducible subject-specific functions. SREA trains the shared branch with joint embedding over augmented temporal neighbors, the routed branch with prototype-based sparse routing and expert specialization, and both branches with numerically stable mutual-guided reweighting based on cross-branch learnability gaps. Experiments on seven public EEG benchmarks across different tasks show that SREA consistently outperforms state-of-the-art methods and EEG foundation models.
\end{abstract}

\section{Introduction}

Electroencephalography (EEG) decoding plays an important role in brain-computer interfaces, computational neuroscience, and neurophysiological assessment~\citep{schalk2004bci2000}. Recent advances in deep learning have substantially improved EEG modeling across tasks such as auditory attention decoding, motor imagery classification, and sleep staging~\citep{schirrmeister2017deep,lawhern2018eegnet,song2022eeg,fan2025listennet}. More recently, large-scale pretraining has led to EEG foundation models, such as LaBraM, EEGPT, CBraMod, and CSBrain, which learn transferable representations from multi-task and multi-dataset corpora~\citep{jianglarge,wang2024eegpt,wangcbramod,zhou2025csbrain}. These models have significantly strengthened generic EEG representation learning and improved transfer across tasks and datasets. However, despite these advances, cross-subject generalization remains a major bottleneck. Models trained on a set of source subjects often degrade markedly when transferred to unseen target subjects~\citep{handiru2016optimized,wang2024dmmr}.

Cross-subject generalization has long been a central challenge in EEG decoding, because neural signals vary substantially across individuals due to differences in brain anatomy, physiology, recording conditions, and behavioral strategies~\citep{handiru2016optimized,wang2024dmmr,bomatter2025limited}. Earlier efforts addressed this problem through handcrafted features, channel selection, and shallow transfer techniques~\citep{handiru2016optimized}. Existing approaches to cross-subject EEG learning largely follow two paradigms. The first learns a shared subject-invariant model, typically by encouraging representations to remove subject-specific variations so that a single classifier can transfer across subjects~\citep{li2018cross,ma2019reducing,shen2022contrastive,wang2024generalizable,wang2024dmmr}. This strategy is effective when the discriminative structure is largely shared, but it can become problematic when subject heterogeneity is strong: forcing all subjects into a single invariant representation may suppress subject-specific but task-relevant cues, leading to over-reduction. The second paradigm uses multiple specialized experts, such as ensembles or mixture-of-experts architectures, to preserve heterogeneous decision patterns across subjects or subject groups~\citep{zhang2024evolutionary,wang2024ensemble,zhao2024ensemble,lisparse}. While this strategy can better capture local specialization, it may rely too heavily on branch-specific structure, fail to extract transferable regularities, and introduce routing uncertainty on unseen subjects, resulting in under-reduction.

In this paper, we formalize cross-subject EEG generalization using algorithmic alignment under distribution shift~\citep{xucan}. Our theoretical analysis shows that a shared expert is better aligned with branch-invariant functions, whereas routed experts are better aligned with branch-dependent functions. More importantly, their relative suitability depends on the reducibility cost of transforming branch-specific functions into branch-invariant ones. Based on this insight, we propose Shared-Routed Expert Alignment (SREA), a collaborative framework that couples a shared branch for reducible invariant structure with routed experts for irreducible subject-specific structure, and bridges them via mutual-guided reweighting driven by cross-branch learnability gaps. Unlike the first paradigm, SREA does not assume all useful variation should be absorbed into one common representation; unlike the second, it does not rely solely on independent expert aggregation. Instead, SREA lets the two branches explicitly supervise each other to dynamically generalize on different cases.

In summary, our contributions are as follows:
\begin{itemize}[leftmargin=*,nosep]
\item We analyze cross-subject EEG generalization from an algorithmic alignment perspective and show that the suitability of shared and routed experts depends on the reducibility cost from branch-specific to branch-invariant functions.
\item We propose SREA, a theory-driven framework that aligns a shared expert with reducible invariant structure and routed experts with irreducible subject-specific structure.
\item We introduce a novel mutual-guided reweighting mechanism that transfers supervision across branches according to learnability gaps, thereby mitigating over-reduction and under-reduction.
\end{itemize}

\section{Theoretical Analysis}
\label{sec:preliminary}

In this section, we formalize cross-subject EEG generalization through algorithmic alignment~\citep{xucan}, aiming to characterize how network architecture influences generalization under subject-induced distribution shift.

\subsection{Architecture Alignment Impacts Cross-Subject Generalization}
\label{sec:thm1}

\begin{definition}
\label{def:alignment}
Let $\mathcal{N}$ denote a neural network composed of $n$ modules $\{\mathcal{N}_i\}_{i=1}^n$, and assume that a target function $g$ can be decomposed into $n$ sub-functions such that $g = f_n \circ \cdots \circ f_1$. We say that $\mathcal{N}$ \emph{aligns} with $g$ if replacing each module $\mathcal{N}_i$ with $f_i$ yields the same output as $g$. The alignment measure is defined as $\mathrm{Align}(\mathcal{N}, g, \epsilon, \delta) := n \cdot \max_{i \in [n]} \mathcal{M}(f_i, \mathcal{N}_i, \epsilon, \delta)$, where $\mathcal{M}(f_i, \mathcal{N}_i, \epsilon, \delta)$ denotes the sample complexity for $\mathcal{N}_i$ to learn $f_i$ with precision $\epsilon$ and failure probability $\delta$.
\end{definition}

Definition \ref{def:alignment} shows that a well-aligned architecture decomposes the problem into simpler sub-problems, reducing sample complexity. We now extend algorithmic alignment~\citep{xucan,lisparse} to domain generalization under distribution shift.

\begin{assumption}
\label{ass:distribution_shift}
Let $p_{\mathrm{tr}}$ and $p_{\mathrm{te}}$ denote the training and test feature distributions at the output of $\mathcal{N}_1$. There exist constants $C \geq 1$ and $\eta \geq 0$ such that for any measurable set $A$, $\mathbb{P}_{D_{\mathrm{te}}}[\mathcal{N}_1(\boldsymbol{x}) \in A] \leq C \cdot \mathbb{P}_{D_{\mathrm{tr}}}[\mathcal{N}_1(\boldsymbol{x}) \in A] + \eta$.
\end{assumption}

\begin{remark}
Subject-invariant physiological patterns ensure a bounded density ratio $C$, while training on diverse subjects yields a reasonable transfer error $\eta$~\citep{berry2017aasm,bomatter2025limited}.
\end{remark}

\begin{assumption}
\label{ass:invariant}
There exists a function $g_c$ such that $g_c(\mathcal{N}_1(\boldsymbol{x})) = y$ for all $\boldsymbol{x} \in \mathcal{E}_{\mathrm{tr}}$, and $\mathbb{P}_{D_{\mathrm{te}}}[\|g_c(\mathcal{N}_1(\boldsymbol{x})) - y\| \leq \epsilon] \geq 1 - \delta$.
\end{assumption}

\begin{remark}
The function $g_c$ represents correlations that persist across both training and test distributions.
\end{remark}

\begin{assumption}
\label{ass:spurious}
There exists a function $g_s$ such that $g_s(\mathcal{N}_1(\boldsymbol{x})) = y$ for all $\boldsymbol{x} \in \mathcal{E}_{\mathrm{tr}}$, and $\mathbb{P}_{D_{\mathrm{te}}}[\|g_s(\mathcal{N}_1(\boldsymbol{x})) - y\| > \omega(\epsilon)] \geq 1 - \delta$.
\end{assumption}

\begin{remark}
The function $g_s$ represents correlations present in training but absent in test distributions.
\end{remark}

\begin{theorem}
\label{thm:backbone}
Let $\mathcal{N}' = \{\mathcal{N}_2, \ldots, \mathcal{N}_n\}$. Suppose we train the neural network with ERM and Assumptions~\ref{ass:distribution_shift}--\ref{ass:spurious} hold. Then:

\noindent (a) If $\mathrm{Align}(\mathcal{N}', g_c, \epsilon, \delta) \leq |\mathcal{E}_{\mathrm{tr}}|$, then $\mathbb{P}_{D_{\mathrm{te}}}[\|\mathcal{N}(\boldsymbol{x}) - y\| \leq O(\epsilon)] \geq 1 - O(\delta) - \eta$.

\noindent (b) If $\mathrm{Align}(\mathcal{N}', g_s, \epsilon, \delta) \leq |\mathcal{E}_{\mathrm{tr}}|$, then $\mathbb{P}_{D_{\mathrm{te}}}[\|\mathcal{N}(\boldsymbol{x}) - y\| > \omega(\epsilon)] \geq 1 - O(\delta) - \eta$.
\end{theorem}

\begin{remark}
Networks aligned with invariant correlations generalize robustly, while those aligned with spurious correlations fail on unseen domains.
\end{remark}

\subsection{Reducibility Cost Impacts the Choice of Shared and Routed Architecture}
\label{sec:thm2}

\begin{definition}
\label{def:branching}
Let $g$ denote a target function with a branching structure $g(\boldsymbol{x}) = \sum_{j=1}^M \mathbf{1}_{I_j}(h_0(\boldsymbol{x})) \cdot h_j(\boldsymbol{x})$, where $h_0$ is a branching function that determines which branch to activate, $\{I_j\}_{j=1}^M$ are disjoint intervals partitioning the range of $h_0$, and $\{h_j\}_{j=1}^M$ are branch-dependent functions. A routed expert network $\mathcal{R}$ consists of $M$ routed experts $\{R_j\}_{j=1}^M$ and a router $R_0$, where each routed expert $R_j$ learns $h_j$ and the router $R_0$ learns $h_0$. A shared expert network $\mathcal{S}$ consists of a single shared expert $S_0$ that learns the entire function $g$.
\end{definition}

\begin{remark}
The branching structure captures branch-dependent processing strategies, which commonly arise under distribution shift due to subject heterogeneity.
\end{remark}

\begin{theorem}
\label{thm:shared_routed}
Let $\mathcal{M}(R_j, h_j, \epsilon, \delta)$ and $\mathcal{M}(R_0, h_0, \epsilon, \delta)$ denote the sample complexity for routed expert $R_j$ to learn $h_j$ and for router $R_0$ to learn $h_0$, respectively, with precision $\epsilon$ and failure probability $\delta$. Let $\mathcal{M}(S_0, g, \epsilon, \delta)$ denote the sample complexity for shared expert $S_0$ to learn $g$. Let $\alpha, \beta$ with $1 \leq \alpha \leq \beta$ satisfy $\alpha \cdot \max_j \mathcal{M}(R_j, h_j, \epsilon, \delta) \leq \mathcal{M}(S_0, g, \epsilon, \delta) \leq \beta \cdot \max_j \mathcal{M}(R_j, h_j, \epsilon, \delta)$. Suppose the branching function is no harder to learn than the branch-dependent functions, i.e., $\mathcal{M}(R_0, h_0, \epsilon, \delta) \leq \max_j \mathcal{M}(R_j, h_j, \epsilon, \delta)$. Then:

\noindent (a) If $\alpha \geq M + 1$, then $\mathrm{Align}(\mathcal{R}, g, \epsilon, \delta) \leq \mathrm{Align}(\mathcal{S}, g, \epsilon, \delta)$.

\noindent (b) If $\beta \leq M + 1$, then $\mathrm{Align}(\mathcal{R}, g, \epsilon, \delta) \geq \mathrm{Align}(\mathcal{S}, g, \epsilon, \delta)$.
\end{theorem}

\begin{remark}
The ratio $\alpha$ (or $\beta$) reflects the reducibility cost of branch-specific functions to a branch-invariant one. When this cost is low (small $\beta$), the shared expert aligns better. When it is high (large $\alpha$), routed experts achieve better alignment.
\end{remark}

\subsection{Collaboration Adapts to Different Cases}
\label{sec:thm3}

\begin{assumption}
\label{ass:complementarity}
For any sample $\boldsymbol{x}$, at least one expert produces a prediction close to the target function, i.e., $\min\{\|S(\boldsymbol{x}) - g(\boldsymbol{x})\|, \|R(\boldsymbol{x}) - g(\boldsymbol{x})\|\} \leq \varepsilon$, where $\varepsilon$ is a small constant.
\end{assumption}

\begin{remark}
This follows from Theorem~\ref{thm:shared_routed}: at least one expert aligns well for each sample.
\end{remark}

\begin{theorem}
\label{thm:error_descent}
Define the alignment errors of the shared expert and routed experts at epoch $t$ as $\epsilon_S^{(t)} = \|S^{(t)}(\boldsymbol{x}) - g(\boldsymbol{x})\|^2$ and $\epsilon_R^{(t)} = \|R^{(t)}(\boldsymbol{x}) - g(\boldsymbol{x})\|^2$, respectively, where $g(\boldsymbol{x})$ denotes the target function. Under Assumption~\ref{ass:complementarity}, the total alignment error satisfies $\mathbb{E}[\epsilon_S^{(t)} + \epsilon_R^{(t)}] \leq \mathbb{E}[\epsilon_S^{(0)} + \epsilon_R^{(0)}]$.
\end{theorem}

\begin{remark}
The expert with higher loss learns from the better-performing one, avoiding over-reduction and under-reduction.
\end{remark}

These results motivate a framework that aligns a shared branch with reducible invariant structure, a routed branch with irreducible subject-specific structure, and a collaborative mechanism that allows the two to correct each other.

\section{Methodology}

\begin{figure*}
    \begin{center}
    \includegraphics[width=\linewidth]{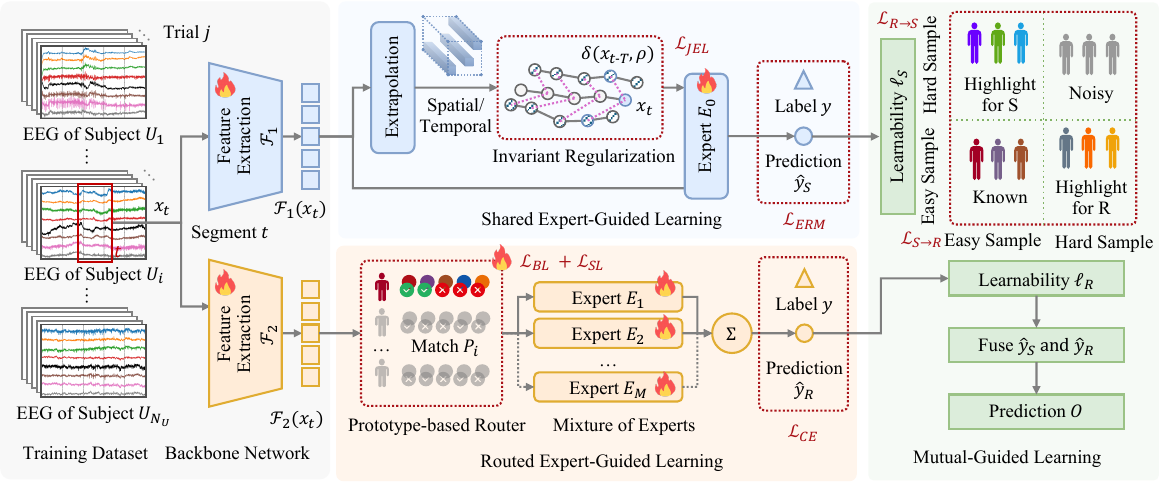}
    \end{center}
    \caption{Overview of the proposed Shared-Routed Expert Alignment (SREA) framework.}
    \label{fig:framework}
\end{figure*}

Fig.~\ref{fig:framework} illustrates the overall framework. Given an EEG segment $\boldsymbol{x}_t$ for subject $U_i$ in trial $j$, SREA extracts representations through two separate feature extractors: $\mathcal{F}_1$ for the shared branch and $\mathcal{F}_2$ for the routed branch. The shared branch trains a shared expert $E_0$ to model stable cross-subject structure using empirical risk minimization together with a joint embedding objective over augmented temporal neighbors. The routed branch uses a router $R$ with prototype-based sparse routing to dispatch samples to specialized routed experts $\{E_j\}_{j=1}^{M}$, while subject-level specialization and load balancing encourage consistent yet non-collapsed expert usage. Finally, the two branches are jointly optimized through mutual-guided reweighting, which emphasizes samples that are learned reliably by one branch but remain difficult for the other. During inference, the predictions of the shared and routed branches ($\hat{y}_S$ and $\hat{y}_R$) are fused to obtain the final output $O$.

\subsection{Shared Branch for Reducible Invariant Structure}
\label{sec:shared_expert_guided_learning}

To satisfy the requirement of Theorem~\ref{thm:shared_routed}, we define a shared expert $E_0$ that models the reducible invariant structure across subjects. Given a batch of EEG signals, we train $E_0$ by minimizing the empirical risk across all subjects using cross-entropy loss:
\begin{equation}
\mathcal{L}_{\text{ERM}} = -\frac{1}{N_B} \sum_{i=1}^{N_B} \sum_{c=1}^{C} y_{ic} \log(\hat{y}_{ic})
\end{equation}
where $N_B$ is the batch size, $C$ is the number of classes, and $y_{ic} \in \{0, 1\}$ is the ground-truth indicator for sample $i$ belonging to class $c$.

To bound the transfer error $\eta$ under distribution shift, we construct augmented temporal neighbors as a proxy for unseen samples whose distributional distance from training data remains controlled. For each EEG sample $\boldsymbol{x}_t \in \mathbb{R}^{E \times L}$, where $E$ is the number of electrodes and $L$ is the segment length, we retrieve its preceding consecutive $T$-second segment $\boldsymbol{x}_{t-T}$ of the same class. We then apply a stochastic spatial masking operation $\delta(\cdot, \rho)$ that randomly masks electrodes with probability $\rho$, yielding $\tilde{\boldsymbol{x}}_{t-T} = \delta(\boldsymbol{x}_{t-T}, \rho)$.

Let $\mathcal{F}_1(\cdot)$ denote the feature extractor of the shared branch. We obtain the representations $\boldsymbol{z}_{t-T} = \mathcal{F}_1(\tilde{\boldsymbol{x}}_{t-T})$ and $\boldsymbol{z}_t = \mathcal{F}_1(\boldsymbol{x}_t)$. The joint embedding loss is defined as:
\begin{equation}
\mathcal{L}_{\text{JEL}} = \frac{1}{N_B} \sum_{i=1}^{N_B} \left( 1 - \frac{\boldsymbol{z}_{i,t-T} \cdot \boldsymbol{z}_{i,t}}{\|\boldsymbol{z}_{i,t-T}\| \|\boldsymbol{z}_{i,t}\|} \right)
\end{equation}
which minimizes the cosine distance between the representations of the current segment and its spatially-masked temporal neighbor.

The total loss for shared expert-guided learning is:
\begin{equation}
\mathcal{L}_{S} = \mathcal{L}_{\text{ERM}} + \mathcal{L}_{\text{JEL}}
\end{equation}

\subsection{Routed Branch for Irreducible Subject-Specific Structure}
\label{sec:routed_expert_guided_learning}

To satisfy the requirement of Theorem~\ref{thm:shared_routed}, we define a routed expert network that captures irreducible subject-specific structure. The routed expert network consists of $M$ routed experts $\{E_j\}_{j=1}^{M}$ and a router $R$. Let $\boldsymbol{z} = \mathcal{F}_2(\boldsymbol{x}) \in \mathbb{R}^{d}$ denote the feature vector extracted from input $\boldsymbol{x}$ by the feature extractor $\mathcal{F}_2$. The output $\boldsymbol{o}$ is computed as a weighted combination of expert outputs:
\begin{equation}
\boldsymbol{o} = \sum_{j=1}^{M} r_{j} \cdot E_j(\boldsymbol{z})
\end{equation}
where $M$ is the number of routed experts, $E_j(\cdot)$ denotes the $j$-th routed expert, and $r_{j}$ represents the routing weight for expert $j$. The routing weights are sparse, with only the top-$K$ experts receiving non-zero weights.

We adopt a prototype-based routing mechanism. Given feature $\boldsymbol{z} \in \mathbb{R}^{d}$, we first project it to a gate dimension space via $W\boldsymbol{z} \in \mathbb{R}^{d_r}$, then compute the routing weights using cosine similarity with learned prototype embeddings $\boldsymbol{P} = [\boldsymbol{P}_1, \ldots, \boldsymbol{P}_{M}] \in \mathbb{R}^{d_r \times M}$:
\begin{equation}
r_{j} = \begin{cases}
\dfrac{\exp(\mathrm{sim}(\boldsymbol{P}_j, W\boldsymbol{z}))}{\sum_{l \in \mathcal{T}} \exp(\mathrm{sim}(\boldsymbol{P}_l, W\boldsymbol{z}))} & \text{if } j \in \mathcal{T} \\[2ex]
0 & \text{otherwise}
\end{cases}
\end{equation}
where $\mathrm{sim}(\boldsymbol{a}, \boldsymbol{b}) = \boldsymbol{a}^{\top} \boldsymbol{b} / (\|\boldsymbol{a}\| \cdot \|\boldsymbol{b}\|)$ denotes the cosine similarity, and $\mathcal{T}$ denotes the set of top-$K$ experts with the highest similarity scores. The prototype embedding matrix $\boldsymbol{P}$ can be interpreted as a codebook for neural activity patterns, and the cosine similarity acts as a matched filter for detecting these patterns. This design mitigates the issue where samples from similar subjects tend to cluster together, enabling each expert to specialize on its assigned cluster.

To optimize the selected experts' predictions toward the ground-truth labels, we employ the cross-entropy loss:
\begin{equation}
\mathcal{L}_{\text{CE}} = -\frac{1}{N_B} \sum_{i=1}^{N_B} \sum_{c=1}^{C} y_{ic} \log(\hat{y}_{ic})
\end{equation}
where $N_B$ is the batch size, $C$ is the number of classes, and $y_{ic} \in \{0, 1\}$ is the ground-truth indicator for sample $i$ belonging to class $c$.

To encourage each expert to specialize on specific subjects, we compute the average routing probability distribution across all samples from each subject, then minimize the entropy of this distribution:
\begin{equation}
\mathcal{L}_{\text{SL}} = \frac{1}{N_U} \sum_{k=1}^{N_U} H_k, \quad H_k = -\sum_{j=1}^{M} \bar{r}_{k,j} \log \bar{r}_{k,j}
\end{equation}
where $N_U$ is the number of unique subjects and $\bar{r}_{k,j} = \frac{1}{|\mathcal{D}_k|} \sum_{\boldsymbol{x} \in \mathcal{D}_k} r_j(\boldsymbol{x})$ is the average routing weight for expert $j$ across all samples from subject $k$, with $\mathcal{D}_k$ denoting the set of samples belonging to subject $k$. This loss encourages the router to consistently route samples from the same subject to the same expert.

To prevent routing collapse where the model consistently selects only a few experts while leaving others undertrained, we employ an auxiliary balance loss:
\begin{equation}
\mathcal{L}_{\text{BL}} = M \sum_{j=1}^{M} f_j \cdot \bar{r}_j
\end{equation}
where $f_j = \frac{1}{N_B} \sum_{i=1}^{N_B} \mathbb{I}[r_j(\boldsymbol{x}_i) > 0]$ denotes the fraction of samples routed to expert $j$ and $\bar{r}_j = \frac{1}{N_B} \sum_{i=1}^{N_B} r_j(\boldsymbol{x}_i)$ represents the average routing weight of expert $j$ over the batch. This loss penalizes configurations where a small subset of experts receives disproportionately high routing weights and selection frequencies.

The total loss for routed expert-guided learning is:
\begin{equation}
\mathcal{L}_{R} = \mathcal{L}_{\text{CE}} + \mathcal{L}_{\text{SL}} + \mathcal{L}_{\text{BL}}
\end{equation}

\subsection{Mutual-Guided Reweighting}
\label{sec:mutual_guided_training}

For a given sample $\boldsymbol{x}^{(i)}$ with ground-truth label $y^{(i)}$, let $\ell_S^{(i)} = -\log P_S(y^{(i)}|\boldsymbol{x}^{(i)})$ and $\ell_R^{(i)} = -\log P_R(y^{(i)}|\boldsymbol{x}^{(i)})$ denote the cross-entropy losses from the shared expert $E_0$ and the routed experts $\{E_j\}_{j=1}^{M}$, respectively. Following Theorem~\ref{thm:error_descent}, the mutual-guided reweighting mechanism couples the two branches by transferring supervision according to cross-branch learnability gaps. Specifically, we derive the reweighting formulas in Eq.~\eqref{eq:r_from_s} and Eq.~\eqref{eq:s_from_r} to ensure that the expert with higher loss learns from the better-performing one, thereby reducing the total alignment error.

When training the routed experts guided by the shared expert, we compute the weighted cross-entropy loss:
\begin{equation}
\label{eq:r_from_s}
\mathcal{L}_{R \leftarrow S} = \frac{1}{N_B} \sum_{i=1}^{N_B} (1+\exp(\ell_R^{(i)} - \ell_S^{(i)})) \cdot \ell_R^{(i)}
\end{equation}
Samples that the shared expert finds easier (lower $\ell_S^{(i)}$) but the routed experts find challenging (higher $\ell_R^{(i)}$) receive higher weights, constraining the routed experts from under-reduction with guidance from the shared expert.

Similarly, when training the shared expert guided by the routed experts:
\begin{equation}
\label{eq:s_from_r}
\mathcal{L}_{S \leftarrow R} = \frac{1}{N_B} \sum_{i=1}^{N_B} (1 + \exp(\ell_S^{(i)} - \ell_R^{(i)})) \cdot \ell_S^{(i)}
\end{equation}
which assigns higher weights to samples that the routed experts handle well, constraining the shared expert from over-reduction.

The total loss for mutual-guided learning combines both directions:
\begin{equation}
\mathcal{L}_{M} = \mathcal{L}_{R \leftarrow S} + \mathcal{L}_{S \leftarrow R}
\end{equation}

\subsection{Inference and Prediction Fusion}

During inference, the shared branch and the routed branch each produce a prediction for the input EEG sample. Let $\hat{y}_S$ and $\hat{y}_R$ denote the predicted class probabilities from the shared expert $E_0$ and the routed experts $\{E_j\}_{j=1}^{M}$, respectively. The final fused output is $O = \frac{1}{2}(\hat{y}_S + \hat{y}_R)$.

\section{Experiments}

\subsection{Experimental Settings}
\label{sec:experimental_setting}

We evaluate the proposed SREA on seven publicly available EEG datasets covering three representative tasks, including motor imagery (MI), sleep stage detection (SSD), and auditory attention decoding (AAD). Specifically, for MI, we use BCI IV-2a~\citep{tangermann2012review} and BCI IV-2b~\citep{leeb2007brain}. For SSD, we use Sleep-EDFx~\citep{goldberger2000physiobank} and ISRUC-3~\citep{khalighi2016isruc}. For AAD, we use the KUL dataset~\citep{das2016effect}, the DTU dataset~\citep{fuglsang2017noise}, and the AVED dataset~\citep{fan2024msfnet}.

We follow the standard experimental protocols used in the corresponding prior studies. For BCI IV-2a and BCI IV-2b, leave-one-subject-out cross-validation is employed~\citep{zhao2024ctnet}. For Sleep-EDFx, subject-wise 10-fold cross-validation is used, while leave-one-subject-out cross-validation is used on ISRUC-3~\citep{wang2024subject}. For the KUL, DTU, and AVED datasets, leave-one-subject-out cross-validation is adopted~\citep{fan2025listennet}. In leave-one-subject-out evaluation, one subject is held out for testing and the remaining subjects are used for training. In subject-wise k-fold cross-validation, one fold of subjects is used for testing and the remaining folds are used for training.

We train both the shared expert and the routed experts using the Adam optimizer. The batch size is set to 256, and the maximum number of training epochs is 100. Early stopping with a patience of 10 epochs is adopted to alleviate overfitting. The learning rate is set to $10^{-4}$, and the weight decay is set to $5 \times 10^{-4}$. Since different tasks involve distinct EEG topologies and temporal resolutions, we adopt task-specific backbones, i.e., DeepConvNet~\citep{schirrmeister2017deep} for MI, SleepWaveNet~\citep{wang2024subject} for SSD, and ListenNet~\citep{fan2025listennet} for AAD. Unless otherwise specified, all loss coefficients are set to 1. The number of routed experts is set to $M = 5$ and one expert is activated for each sample. The masking ratio is set to $\rho = 10\%$. For datasets with fewer than 10 electrodes, temporal masking is used instead of electrode masking.

\begin{table*}[t]
\caption{Comparison under subject-independent conditions. Accuracy (\%) is reported. $^\ast$ indicates significance (paired $t$-test, $p < 0.05$).}
\label{tab:comparison_experiments}
\begin{center}
\adjustbox{max width=\linewidth}{
\begin{tabular}{@{}lcc@{}}
\toprule
\multicolumn{3}{c}{(a) Motor Imagery} \\
\midrule
Model & IV-2a & IV-2b \\
\midrule
ShallowConvNet & 56.8$^\ast$ & 74.3$^\ast$ \\
DeepConvNet    & 60.2$^\ast$ & 75.2$^\ast$ \\
EEGNet         & 56.9$^\ast$ & 75.1$^\ast$ \\
Conformer      & 53.4$^\ast$ & 73.5$^\ast$ \\
CTNet          & 58.6$^\ast$ & 76.3$^\ast$ \\
\midrule
Proposed       & \textbf{61.3} & \textbf{77.5} \\
\bottomrule
\end{tabular}
\hspace{1em}
\begin{tabular}{@{}lcc@{}}
\toprule
\multicolumn{3}{c}{(b) Sleep Stage Detection} \\
\midrule
Model & EDFx & ISRUC \\
\midrule
DeepSleepNet    & 81.8$^\ast$ & 76.5$^\ast$ \\
Utime           & 80.6$^\ast$ & 73.5$^\ast$ \\
TinySleepNet    & 82.7$^\ast$ & 75.8$^\ast$ \\
SalientSleepNet & 82.9$^\ast$ & 76.9$^\ast$ \\
SleepWaveNet    & 83.2$^\ast$ & 79.2$^\ast$ \\
\midrule
Proposed        & \textbf{84.1} & \textbf{79.8} \\
\bottomrule
\end{tabular}
\hspace{1em}
\begin{tabular}{@{}lccc@{}}
\toprule
\multicolumn{4}{c}{(c) Auditory Attention Decoding} \\
\midrule
Model & KUL & DTU & AVED \\
\midrule
SSF-CNN   & 59.3$^\ast$ & 52.3$^\ast$ & 51.7$^\ast$ \\
MBSSFCC   & 62.7$^\ast$ & 52.5$^\ast$ & 52.2$^\ast$ \\
DBPNet    & 61.1$^\ast$ & 55.5$^\ast$ & 52.1$^\ast$ \\
DARNet    & 69.9$^\ast$ & 55.6$^\ast$ & 51.3$^\ast$ \\
ListenNet & 78.1$^\ast$ & 56.8$^\ast$ & 52.8$^\ast$ \\
\midrule
Proposed  & \textbf{81.7} & \textbf{58.3} & \textbf{53.9} \\
\bottomrule
\end{tabular}
}
\end{center}
\end{table*}

\subsection{Comparison with Task-Specific Methods}
\label{sec:comparison_experiments_task_specific}

We compare with several representative task-specific models on the motor imagery task using the BCI IV-2a and BCI IV-2b datasets. The compared methods include ShallowConvNet~\citep{schirrmeister2017deep}, DeepConvNet~\citep{schirrmeister2017deep}, EEGNet~\citep{lawhern2018eegnet}, Conformer~\citep{song2022eeg}, and CTNet~\citep{zhao2024ctnet}. The results in Table~\ref{tab:comparison_experiments} show that the proposed method achieves the highest accuracy, obtaining 61.3\% on BCI IV-2a and 77.5\% on BCI IV-2b. Notably, the gain is achieved under a highly challenging setting with only nine subjects available for training and testing, indicating that the proposed framework can effectively alleviate cross-subject discrepancy even when the number of training subjects is limited.

We then evaluate the proposed method on sleep stage detection using the Sleep-EDFx and ISRUC-3 datasets. We compare with strong baselines, including DeepSleepNet~\citep{supratak2017deepsleepnet}, Utime~\citep{perslev2019u}, TinySleepNet~\citep{supratak2020tinysleepnet}, SalientSleepNet~\citep{jia2021salientsleepnet}, and SleepWaveNet~\citep{wang2024subject}. As shown in Table~\ref{tab:comparison_experiments}, our method obtains the best performance on both datasets, achieving 84.1\% on Sleep-EDFx and 79.8\% on ISRUC-3. Although the improvements are relatively smaller than those on MI and AAD, the proposed framework still provides consistent gains. A possible reason is that SSD usually contains stronger subject-shared physiological patterns, such as sleep spindles and K-complexes, making the cross-subject gap less severe than in other EEG tasks.

Next, we evaluate auditory attention decoding. The compared methods include SSF-CNN~\citep{cai2021low}, MBSSFCC~\citep{jiang2022detecting}, DBPNet~\citep{ni2024dbpnet}, DARNet~\citep{yan2024darnet}, and ListenNet~\citep{fan2025listennet}. As reported in Table~\ref{tab:comparison_experiments}, the proposed method achieves the best performance on all three datasets, reaching 81.7\% on KUL, 58.3\% on DTU, and 53.9\% on AVED. These results verify the advantage of the proposed framework for cross-subject auditory attention decoding, even under short 1-second decision windows.

\subsection{Comparison with EEG Foundation Methods}
\label{sec:comparison_experiments_eeg_foundation}

We further evaluate the proposed framework in a multi-task multi-dataset setting under subject-independent protocols. Following prior work~\citep{zhou2025csbrain}, we report balanced accuracy and compare with recent EEG foundation models, including FFCL, ST-Trans, BIOT, LaBraM, and CBraMod. We use CBraMod as the backbone for our method. As shown in Table~\ref{tab:comparison_pretrained}, foundation models already provide strong cross-dataset generalization, while the proposed framework further improves the best-performing baseline across most datasets. This result suggests that SREA is not limited to lightweight task-specific architectures, but can also serve as a general collaboration paradigm that complements powerful pre-trained EEG representations.

\begin{table*}[t]
\caption{Comparison with EEG foundation models on multi-task multi-dataset scenarios under subject-independent conditions. Balanced accuracy (\%) is reported. $^\ast$ indicates significance (paired $t$-test, $p < 0.05$).}
\label{tab:comparison_pretrained}
\begin{center}
\begin{tabular}{@{}lcc@{}}
\toprule
\multicolumn{3}{c}{(a) Motor Imagery} \\
\midrule
Model & IV-2a & IV-2b \\
\midrule
FFCL       & 44.7$^\ast$ & 74.8$^\ast$ \\
ST-Trans   & 45.8$^\ast$ & 75.0$^\ast$ \\
BIOT       & 47.5$^\ast$ & 76.0$^\ast$ \\
LaBraM     & 48.7$^\ast$ & 76.4$^\ast$ \\
CbraMod    & 51.4$^\ast$ & 76.9$^\ast$ \\
\midrule
Proposed   & \textbf{56.2} & \textbf{77.6} \\
\bottomrule
\end{tabular}
\hspace{1em}
\begin{tabular}{@{}lcc@{}}
\toprule
\multicolumn{3}{c}{(b) Sleep Stage Detection} \\
\midrule
Model & EDFx & ISRUC \\
\midrule
FFCL       & 77.3$^\ast$ & 72.8$^\ast$ \\
ST-Trans   & 77.9$^\ast$ & 73.8$^\ast$ \\
BIOT       & 80.8$^\ast$ & 75.3$^\ast$ \\
LaBraM     & 81.7$^\ast$ & 76.3$^\ast$ \\
CbraMod    & 82.1$^\ast$ & 78.7$^\ast$ \\
\midrule
Proposed   & \textbf{83.2} & \textbf{79.6} \\
\bottomrule
\end{tabular}
\hspace{1em}
\begin{tabular}{@{}lccc@{}}
\toprule
\multicolumn{4}{c}{(c) Auditory Attention Decoding} \\
\midrule
Model & KUL & DTU & AVED \\
\midrule
FFCL       & 73.4$^\ast$ & 52.0$^\ast$ & 51.4$^\ast$ \\
ST-Trans   & 73.6$^\ast$ & 52.4$^\ast$ & 51.5$^\ast$ \\
BIOT       & 74.8$^\ast$ & 54.6$^\ast$ & 52.5$^\ast$ \\
LaBraM     & 75.5$^\ast$ & 55.3$^\ast$ & 52.9$^\ast$ \\
CbraMod    & 76.5$^\ast$ & 55.6$^\ast$ & 53.8$^\ast$ \\
\midrule
Proposed   & \textbf{79.8} & \textbf{57.0} & \textbf{53.9} \\
\bottomrule
\end{tabular}
\end{center}
\end{table*}

\subsection{Case Studies}

To further understand when the proposed framework is effective, we conduct case studies on the KUL dataset. We follow the same experimental settings as in Sec.~\ref{sec:comparison_experiments_task_specific} and report the subject-wise test accuracy of the full model, together with two ablated variants: using only the shared expert (w/o R) and using only the routed experts (w/o S). The results are shown in Fig.~\ref{fig:case_study}.

\begin{figure*}
    \begin{center}
    \includegraphics[width=\linewidth]{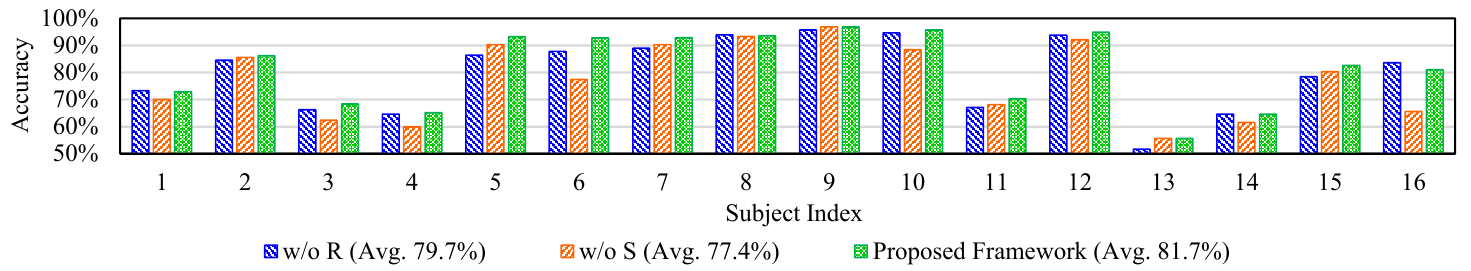}
    \end{center}
    \caption{Case studies of the proposed framework on the KUL dataset. The x-axis represents the subject index, and the y-axis represents the accuracy.}
    \label{fig:case_study}
\end{figure*}

It can be observed that the proposed framework outperforms the two ablated variants on nearly all subjects and achieves the best average accuracy overall. Specifically, the full model surpasses the shared-only and routed-only variants for all subjects except Subject 16. This result provides direct evidence that the two branches capture complementary information. The shared expert improves generalization by focusing on subject-invariant patterns, while the routed experts preserve subject-specific discriminative cues. Their collaboration therefore mitigates both over-reduction and under-reduction, leading to more reliable predictions across subjects.

\subsection{Ablation Study}

To validate the contribution of each component, we conduct systematic ablation experiments on the KUL and DTU datasets. We consider two levels of ablation: (1) branch-level variants that remove entire architectural components, and (2) component-level variants that disable individual loss terms within the full model. The results are reported in Table~\ref{tab:ablation}.

\begin{table}[t]
\caption{Ablation study on the KUL and DTU datasets. Average accuracy (\%) is reported under leave-one-subject-out cross-validation. \emph{w/o Routed}: shared branch only. \emph{w/o Shared}: routed branch only. \emph{w/o Mutual}: both branches trained independently and fused at inference. \emph{w/o JEL}: no joint embedding loss. \emph{w/o SL+BL}: no expert routing losses.}
\label{tab:ablation}
\begin{center}
\begin{tabular}{lcc@{\hspace{2em}}lcc}
\toprule
\multicolumn{3}{c}{(a) Branch-level} & \multicolumn{3}{c}{(b) Component-level} \\
\cmidrule(lr){1-3} \cmidrule(lr){4-6}
Variant & KUL & DTU & Variant & KUL & DTU \\
\midrule
Proposed   & \textbf{81.7} & \textbf{58.3} & Proposed  & \textbf{81.7} & \textbf{58.3} \\
w/o Routed & 79.4 & 57.2 & w/o JEL   & 81.1 & 57.9 \\
w/o Shared & 79.9 & 57.0 & w/o SL+BL & 81.0 & 57.7 \\
w/o Mutual & 80.8 & 57.8 &           &      &      \\
\bottomrule
\end{tabular}
\end{center}
\end{table}

Several observations can be made from Table~\ref{tab:ablation}. First, removing either branch causes the largest performance drops on both datasets, confirming that both the shared and routed branches are essential. On KUL, the routed branch contributes more, indicating that subject-specific expert specialization is important when more subjects are available for routing. On DTU, both branches contribute comparably, with the shared branch showing a slightly larger individual contribution. Second, the w/o Mutual variant shows that naive ensembling without mutual-guided reweighting consistently underperforms the full model, demonstrating that cross-branch collaboration during training provides gains beyond simple prediction fusion. Third, both component-level ablations show consistent drops, verifying that the joint embedding loss stabilizes shared representations and that the routing losses prevent expert collapse and encourage meaningful specialization.

\subsection{Empirical Verification of Theoretical Assumptions}

To validate that the assumptions underlying Theorems~\ref{thm:backbone}--\ref{thm:error_descent} hold on real EEG data, we conduct quantitative verification experiments on all three benchmark datasets. For each dataset, we train a single-fold subject-independent model and compute three metrics tied to the theoretical assumptions.

\begin{table}[t]
\caption{Empirical verification of theoretical assumptions on real EEG datasets. MMD is computed with multi-scale RBF kernels between training and test features. CKA measures the representational similarity between the shared and routed branches. The generalization gap is defined as the difference between training accuracy and test accuracy.}
\label{tab:theory_verification}
\begin{center}
\adjustbox{max width=\linewidth}{
\begin{tabular}{@{}lccc@{\hspace{2em}}lccc@{}}
\toprule
\multicolumn{4}{c}{(a) Assumption~\ref{ass:distribution_shift}: Bounded distribution shift} & \multicolumn{4}{c}{(b) Assumptions~\ref{ass:invariant}\&\ref{ass:spurious}: Correlations} \\
\cmidrule(lr){1-4} \cmidrule(lr){5-8}
Metric & KUL & IV-2a & ISRUC & Metric & KUL & IV-2a & ISRUC \\
\midrule
Raw MMD ($\times 10^{-3}$)    & 1.47 & 1.41 & 1.13 & CKA (test set)      & 0.737 & 0.623 & 0.856 \\
Shared MMD ($\times 10^{-3}$) & 21.5 & 3.59 & 6.54 & Gap (shared branch)  & 0.220 & 0.170 & 0.021 \\
Routed MMD ($\times 10^{-3}$) & 12.7 & 4.88 & 3.90 & Gap (routed branch)  & 0.231 & 0.188 & 0.035 \\
\bottomrule
\end{tabular}
}
\end{center}
\end{table}

As shown in Table~\ref{tab:theory_verification}, the raw-feature MMD values are small on all three datasets, confirming that the cross-subject distribution shift required by Assumption~\ref{ass:distribution_shift} is bounded. CKA between shared and routed branch representations on held-out test subjects is consistently high, supporting the existence of invariant correlations posited by Assumption~\ref{ass:invariant}. The generalization gap of the routed branch exceeds that of the shared branch on all datasets, indicating that routed experts capture more subject-specific correlations that do not transfer, consistent with Assumption~\ref{ass:spurious}.

\subsection{Limitation}

The proposed framework has been evaluated on classification tasks. Extending it to regression-based EEG decoding tasks, such as vigilance estimation~\citep{zheng2017vigilance} and speech envelope reconstruction~\citep{defossez2023speech}, remains an open direction for future work.

\section{Conclusion}

In this paper, we show that shared and routed experts are suited to different components of the target function, with their relative advantage determined by the reducibility cost from branch-specific to branch-invariant structure. Based on this insight, we proposed SREA. Experiments under subject-independent settings show that SREA consistently outperforms strong task-specific baselines, standard ensemble methods, and recent EEG foundation models. We further conduct case studies, ablations, complexity analysis, and validation on assumptions, all of which support the proposed design and its underlying alignment principle.

\bibliographystyle{plainnat}
\bibliography{neurips_2026}

\newpage
\appendix

\section{Appendix}

This document accompanies the paper \emph{Aligning Shared and Routed Experts for Cross-Subject EEG Generalization}.
We provide benchmark dataset descriptions, implementation details, comparison with ensemble methods, synthetic verification, qualitative analysis, and complete proofs for the theorems.

\subsection{Benchmark Datasets}
\label{sec:benchmark_datasets}

In this section, we detail the benchmark datasets and preprocessing used in the main text. We evaluate the proposed framework on seven publicly available datasets spanning three EEG analysis tasks. The statistics of all datasets are summarized in Table~\ref{tab:experiment_statistics}.

For preprocessing, we apply a bandpass filter with cutoff frequencies of 0.1--50 Hz to preserve task-relevant neural activity, followed by a 50 Hz notch filter to suppress power-line interference. For BCI IV-2a and BCI IV-2b, the signals are segmented into 4-second windows without overlap. For Sleep-EDFx and ISRUC-3, we use 30-second windows without overlap. For the KUL, DTU, and AVED datasets, EEG signals are segmented into 1-second decision windows with 50\% overlap.

\begin{table}[h]
\caption{Statistics of seven publicly available datasets. We report the task type, total number of subjects, total number of categories, and total number of electrodes. AAD denotes auditory attention decoding, MI denotes motor imagery, and SSD denotes sleep stage detection.}
\label{tab:experiment_statistics}
\begin{center}
\begin{tabular}{llccc}
    \toprule
    Dataset           & Task & Subjects & Categories & Electrodes  \\
    \midrule
    BCI IV-2a (2012)  & MI   & 9     & 4     & 22  \\
    BCI IV-2b (2007)  & MI   & 9     & 2     & 3   \\
    Sleep-EDFx (2000) & SSD  & 79    & 5     & 2   \\
    ISRUC-3 (2016)    & SSD  & 9     & 5     & 6   \\
    KUL (2016)        & AAD  & 16    & 2     & 64  \\
    DTU (2017)        & AAD  & 18    & 2     & 64  \\
    AVED (2024)       & AAD  & 20    & 2     & 32  \\
    \bottomrule
\end{tabular}
\end{center}
\end{table}

For motor imagery classification, we use the BCI IV-2a dataset \citep{tangermann2012review} and the BCI IV-2b dataset \citep{leeb2007brain}.

The BCI IV-2a dataset contains EEG recordings from 9 subjects (A01--A09). EEG signals were recorded using 22 Ag/AgCl electrodes at a sampling rate of 250 Hz. Each subject participated in two sessions on separate days, with the first session allocated for training and the second for testing. Each session consisted of 288 trials, with 72 trials per class. We used the temporal segment from 2 to 6 seconds in our experiments. Labels indicate one of four motor imagery tasks, including left hand, right hand, both feet, and tongue.

The BCI IV-2b dataset contains EEG recordings from 9 subjects (B01--B09). EEG signals were recorded using 3 bipolar channels (C3, Cz, and C4) at a sampling rate of 250 Hz. Each subject participated in five sessions, where the first three sessions were used for calibration and the remaining two sessions were used for testing. There are approximately 400 trials in the training set and 320 trials in the test set. We used the temporal segment from 3 to 7 seconds in our experiments. Labels indicate one of two motor imagery tasks, including left hand and right hand.

For sleep stage classification, we use the Sleep-EDFx dataset \citep{goldberger2000physiobank} and the ISRUC-3 dataset \citep{khalighi2016isruc}.

The Sleep-EDFx dataset contains whole-night polysomnographic recordings from 78 healthy individuals aged 25 to 101 years. Each recording includes EEG, EOG, and chin EMG signals. Each participant has two day-night recordings, with the exception of three individuals who have only one recording due to equipment issues. The dataset comprises 195,479 sleep epochs in total. Labels indicate one of five sleep stages annotated in 30-second epochs, including Wake, N1, N2, N3, and REM.

The ISRUC-3 dataset contains polysomnographic recordings from 10 healthy subjects. Each recording includes EEG, EOG, and EMG signals, segmented into 30-second epochs. We use data from 30 minutes before and after the in-bed period. Following the latest AASM sleep standard, we merge the N3 and N4 stages into a single N3 stage. Labels indicate one of five sleep stages, including Wake, N1, N2, N3, and REM.

For auditory attention decoding, we use the KUL dataset \citep{das2016effect}, the DTU dataset \citep{fuglsang2017noise}, and the AVED dataset \citep{fan2024msfnet}.

The KUL dataset contains EEG recordings from 16 normal-hearing subjects. EEG signals were recorded using a 64-channel BioSemi ActiveTwo system. Each subject was instructed to attend to one of two competing speech streams, with speakers positioned at 90° to the left or right. Each subject completed 8 trials, each lasting 6 minutes. Labels indicate the attended direction.

The DTU dataset contains EEG recordings from 18 normal-hearing subjects. EEG signals were recorded using a 64-channel BioSemi system. Each subject performed a target speaker tracking task in an environment with reverberation and dynamic background noise, attending to one of two competing speakers positioned at 60° relative to the subject. Each subject completed 60 trials, each lasting approximately 50 seconds. Labels indicate the attended speaker.

The AVED dataset contains EEG recordings from 20 normal-hearing subjects. EEG signals were recorded using a 32-channel system. Subjects were evenly divided into two experimental conditions: audio-only and audio-visual, with 10 subjects in each condition. Each subject was instructed to attend to one of two competing speech streams, with speakers positioned at 90° to the left or right. In the audio-visual condition, subjects also watched a video of the narrator they were instructed to focus on. Each subject completed 16 trials, each lasting 152 seconds. Labels indicate the attended direction.

\subsection{Implementation Details}
\label{sec:implementation_details}

We use the Adam optimizer for both the shared expert and the routed experts. The batch size $N_B$ is set to 256, and the maximum number of training epochs is set to 100. Early stopping with a patience of 10 epochs is employed to prevent overfitting. The learning rate is set to 0.0001 for the shared expert and the routed experts. Weight decay is set to 0.0005. Since task-specific models cannot accommodate different tasks with varying topologies and decision windows, we use ListenNet \citep{fan2025listennet} for AAD, DeepConvNet \citep{schirrmeister2017deep} for MI, and SleepWaveNet \citep{wang2024subject} for SSD. We adopt these network architectures but do not use the training methods from these papers. Instead, we use the proposed SREA for training.

We conduct experiments with foundation brain models, including LaBraM \citep{jianglarge}, CBraMod \citep{wangcbramod}, and EEGPT \citep{wang2024eegpt}. To adapt these foundation brain models for AAD, MI, and SSD, we prepend two layers before the foundation brain models: a 1D convolutional layer that maps the electrode channels from the dataset to the number of channels used during pretraining, and a resampling layer that resamples the EEG signals to match the patch size used in pretraining. We replace the classification head with the proposed shared expert and routed experts.

The proposed loss function is insensitive to the weighting coefficients, so all coefficients are set to 1. We use $M = 5$ experts and activate one expert at a time. We set $\rho = 10\%$. When the number of electrodes is fewer than 10, we randomly mask temporal segments within the electrodes. We conduct experiments on 2 Tesla A100 GPUs.

The source code is implemented based on PyTorch \citep{paszke2019pytorch} and TorchEEG \citep{zhang2024torcheegemo}. We provide the code in the supplementary materials for reviewers to reproduce all experimental results. We will release the code on GitHub upon paper acceptance for follow-up studies.

\subsection{Comparison with Ensemble Methods}
\label{sec:ensemble_comparison}

To further demonstrate that the gains of the proposed framework arise from collaborative shared-routed alignment rather than simple model aggregation, we compare with two standard ensemble strategies: Bagging and Boosting. Both ensemble baselines are built on the same task-specific backbone used by the proposed method (ListenNet for AAD, DeepConvNet for MI, and SleepWaveNet for SSD). Table~\ref{tab:ensemble} reports the results on all seven datasets.

\begin{table*}[h]
\caption{Comparison with ensemble methods under subject-independent conditions. Accuracy (\%) is reported. $^\ast$ indicates significance (paired $t$-test, $p < 0.05$).}
\label{tab:ensemble}
\begin{center}
\adjustbox{max width=\linewidth}{
\begin{tabular}{@{}lccc@{}}
\toprule
\multicolumn{4}{c}{(a) Auditory Attention Decoding} \\
\midrule
Model & KUL & DTU & AVED \\
\midrule
ListenNet            & 78.1$^\ast$ & 56.8$^\ast$ & 52.8$^\ast$ \\
Bagging (ListenNet)  & 79.1$^\ast$ & 57.2$^\ast$ & 53.2$^\ast$ \\
Boosting (ListenNet) & 79.6$^\ast$ & 57.5$^\ast$ & 53.1$^\ast$ \\
\midrule
Proposed             & \textbf{81.7} & \textbf{58.3} & \textbf{53.9} \\
\bottomrule
\end{tabular}
\hspace{1em}
\begin{tabular}{@{}lcc@{}}
\toprule
\multicolumn{3}{c}{(b) Motor Imagery} \\
\midrule
Model & IV-2a & IV-2b \\
\midrule
DeepConvNet            & 60.2$^\ast$ & 75.2$^\ast$ \\
Bagging (DeepConvNet)  & 60.6$^\ast$ & 76.6$^\ast$ \\
Boosting (DeepConvNet) & 60.5$^\ast$ & 76.8$^\ast$ \\
\midrule
Proposed               & \textbf{61.3} & \textbf{77.5} \\
\bottomrule
\end{tabular}
\hspace{1em}
\begin{tabular}{@{}lcc@{}}
\toprule
\multicolumn{3}{c}{(c) Sleep Stage Detection} \\
\midrule
Model & EDFx & ISRUC \\
\midrule
SleepWaveNet            & 83.2$^\ast$ & 79.2$^\ast$ \\
Bagging (SleepWaveNet)  & 83.3$^\ast$ & 79.4$^\ast$ \\
Boosting (SleepWaveNet) & 83.5$^\ast$ & 79.4$^\ast$ \\
\midrule
Proposed                & \textbf{84.1} & \textbf{79.8} \\
\bottomrule
\end{tabular}
}
\end{center}
\end{table*}

As shown in Table~\ref{tab:ensemble}, although Bagging and Boosting provide modest improvements over the corresponding single-backbone baselines, the proposed framework consistently outperforms both ensemble strategies across all seven datasets. For example, on KUL, the proposed method achieves 81.7\%, compared to 79.6\% for Boosting and 79.1\% for Bagging. On BCI IV-2b, the proposed method reaches 77.5\%, exceeding both Bagging (76.6\%) and Boosting (76.8\%). These results confirm that the gains of SREA stem from the principled collaboration between shared and routed experts via mutual-guided reweighting, rather than from simple model aggregation. Standard ensemble methods aggregate independently trained models without distinguishing reducible invariant structure from irreducible subject-specific structure, limiting their ability to address the complementary failure modes identified in our theoretical analysis.

\subsection{Synthetic Verification}
\label{sec:synthetic_verification}

To empirically validate the theoretical analysis in Sec.~\ref{sec:preliminary}, we construct a synthetic domain generalization benchmark with hierarchical domain structures. Specifically, we organize domains into groups to simulate subject groups with similar neural patterns. Each group center is first sampled from a Gaussian distribution, and each domain center is then obtained by perturbing its corresponding group center. Samples are finally drawn from Gaussian distributions centered at the domain-specific means.

Based on this construction, we define a domain-independent function shared across all domains and a domain-dependent function generated by hierarchical perturbations of group-level parameters. The target function of each domain is formulated as a weighted combination of these two components, where the weighting coefficient controls the proportion of reducible components.

We use five source domains, each containing 500 samples with an input dimension of 128, and conduct leave-one-domain-out cross-validation. By varying the proportion of reducible components from 0 to 1, we evaluate the performance of the full model and its ablated variants. The results are shown in Fig.~\ref{fig:synthetic_experiments}.

\begin{figure}[h]
    \begin{center}
    \includegraphics[width=0.6\linewidth]{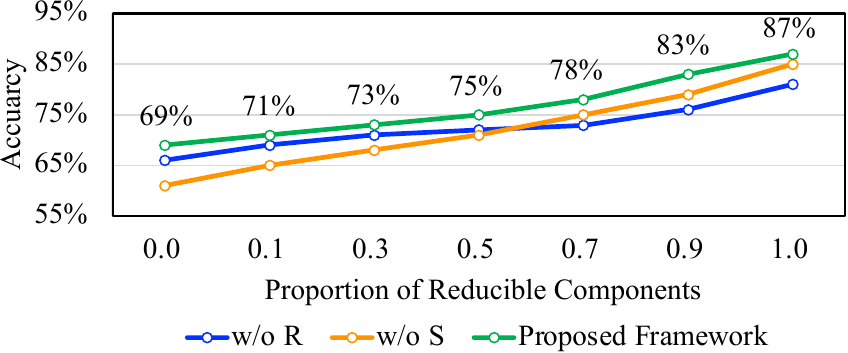}
    \end{center}
    \caption{Experiments on synthetic verification under varying proportions of reducible components. We report the target domain test accuracy (\%) with different ablation settings.}
    \label{fig:synthetic_experiments}
\end{figure}

Several observations can be made. When the proportion of reducible components is low, the routed experts perform better than the shared expert, since domain-specific processing dominates. As the reducible proportion increases, the advantage of routed experts gradually diminishes, while the shared expert becomes increasingly effective. Importantly, the proposed full framework maintains the best performance across different settings, indicating that it can adaptively benefit from both paradigms. This result strongly supports our theoretical claim that the effectiveness of shared and routed experts depends on the reducibility cost of domain-specific functions to domain-invariant ones.

\subsection{Qualitative Analysis}
\label{sec:qualitative_analysis}

We further perform qualitative analysis to visualize the representations learned by the shared expert and routed experts. Specifically, we train the model on the BCI IV-2b dataset and use UMAP to project the learned embeddings of both training and test samples into a two-dimensional space. The visualization results are presented in Fig.~\ref{fig:qualitative_experiment}.

\begin{figure}
    \begin{center}
    \includegraphics[width=0.6\linewidth]{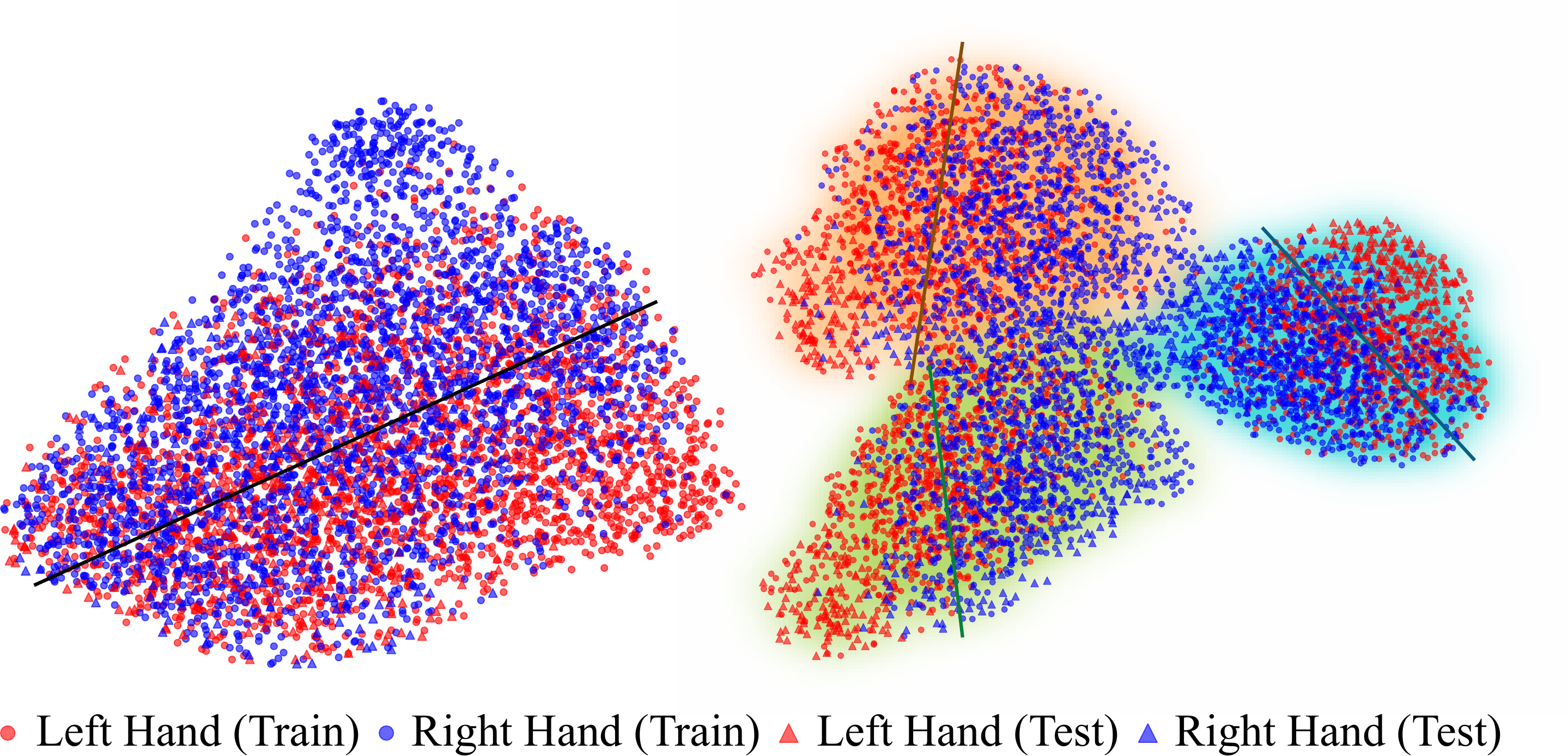}
    \end{center}
    \caption{Qualitative analysis on the BCI-IV-2b dataset. The left subfigure visualizes the embeddings from the shared expert model using UMAP for the 1st fold of the training and test sets, and the right subfigure shows the embeddings from the routed expert model using UMAP for the 1st fold of the training and test sets.}
    \label{fig:qualitative_experiment}
\end{figure}

As shown in the left panel of Fig.~\ref{fig:qualitative_experiment}, the embeddings produced by the shared expert exhibit a relatively unified decision boundary, indicating that the shared branch captures common task-relevant structures across subjects. In contrast, the right panel shows that the routed experts form multiple specialized decision regions. Moreover, the routed assignments reveal that different samples are adaptively handled by different experts, suggesting that the routed branch indeed models heterogeneous subject-specific patterns. These observations are consistent with our theoretical analysis: the shared expert is suitable for reducible and invariant components, whereas the routed experts are advantageous for irreducible and branch-dependent components.

\subsection{Complexity Analysis}

To demonstrate the practical efficiency of the proposed framework, we analyze the parameter overhead and per-batch training time introduced by each component. Fig.~\ref{fig:complexity} visualizes both metrics on three representative datasets using the corresponding task-specific backbones.

\begin{figure}[h]
    \begin{center}
    \includegraphics[width=0.85\linewidth]{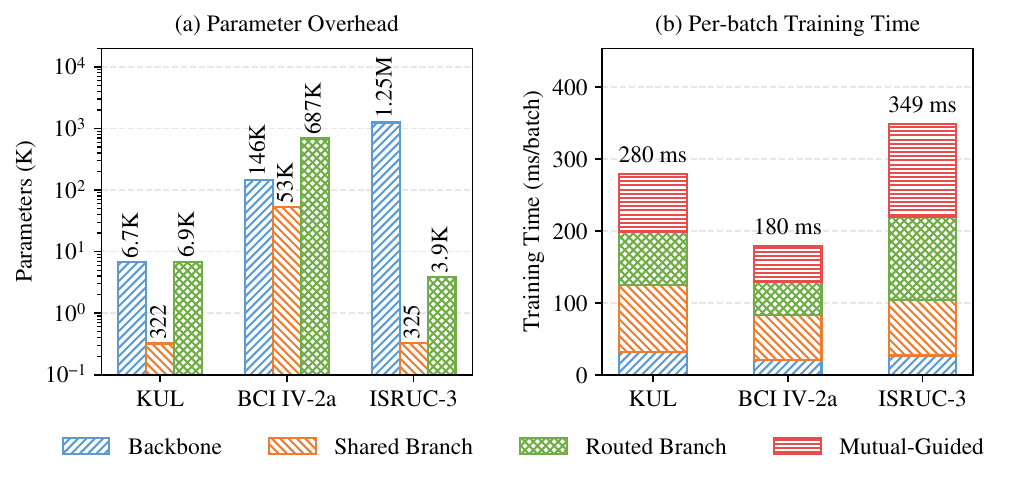}
    \end{center}
    \vspace{-1em}
    \caption{Complexity analysis. (a) Parameter overhead per component (log scale). The mutual-guided reweighting introduces zero extra parameters. (b) Per-batch training time breakdown (batch size 256, NVIDIA A100).}
    \label{fig:complexity}
\end{figure}

As shown in Fig.~\ref{fig:complexity}, the parameter overhead of the routed MoE head scales with the backbone's feature dimension, ranging from 3.9K (+0.3\%) for SleepWaveNet to 686.6K for DeepConvNet, while the mutual-guided mechanism introduces zero additional parameters. In terms of training time, each component adds moderate per-batch cost, and the total per-batch time ranges from 180 ms to 349 ms across different backbones.

\subsection{Proof for Theorems}

In this section, we provide complete proofs for the theorems stated in the main text. We first state the algorithmic alignment theorem in the i.i.d.\ setting \citep{xucan}.

\begin{theorem}
\label{thm:iid_alignment}
Fix $\epsilon > 0$ and $\delta \in (0,1)$. Let $g$ be a target function and $\mathcal{N}$ be a neural network with $n$ modules. Suppose $\{\boldsymbol{x}_i\}_{i=1}^M$ are i.i.d.\ samples drawn from a distribution $D$, and let $y_i = g(\boldsymbol{x}_i)$. Assume the following conditions hold:

\noindent (a) Algorithmic Stability. Let $A$ be a learning algorithm for each module $\mathcal{N}_i$. Suppose $f = A(\{(\boldsymbol{x}_i, y_i)\})$ and $\hat{f} = A(\{(\hat{\boldsymbol{x}}_i, y_i)\})$. Then for any $\boldsymbol{x}$, $\|f(\boldsymbol{x}) - \hat{f}(\boldsymbol{x})\| \leq L_0 \cdot \max_i \|\boldsymbol{x}_i - \hat{\boldsymbol{x}}_i\|$.

\noindent (b) Sequential Training. The modules $\mathcal{N}_1, \ldots, \mathcal{N}_n$ are trained sequentially. Module $\mathcal{N}_1$ is trained on samples $\{(\boldsymbol{x}_i^{(1)}, f_1(\boldsymbol{x}_i^{(1)}))\}_{i=1}^M$, where $\boldsymbol{x}_i^{(1)}$ are obtained from the training dataset. For $j > 1$, the inputs $\hat{\boldsymbol{x}}_i^{(j)}$ to module $\mathcal{N}_j$ are the outputs from the previous modules, while the labels are generated by the ground-truth functions $f_1, \ldots, f_{j-1}$ applied to $\boldsymbol{x}_i^{(1)}$.

\noindent (c) Lipschitz Continuity. Each learned function $\hat{f}_j$ satisfies $\|\hat{f}_j(\boldsymbol{x}) - \hat{f}_j(\hat{\boldsymbol{x}})\| \leq L_1 \|\boldsymbol{x} - \hat{\boldsymbol{x}}\|$ for some constant $L_1 > 0$.

\noindent Under conditions (a)--(c), if $\mathrm{Align}(\mathcal{N}, g, \epsilon, \delta) \leq M$, then there exists a learning algorithm $A$ such that $\mathbb{P}_{\boldsymbol{x} \sim D}[\|\mathcal{N}(\boldsymbol{x}) - g(\boldsymbol{x})\| \leq O(\epsilon)] \geq 1 - O(\delta)$, where $\mathcal{N}$ denotes the network produced by $A$ on the training data $\{(\boldsymbol{x}_i, y_i)\}_{i=1}^M$.
\end{theorem}

Next, we state a lemma that transfers probability bounds from the training distribution to the test distribution, given Assumption 1.

\begin{lemma}
\label{lem:prob_transfer}
Under Assumption 1, for any function $f$ and measurable set $A$, if $\mathbb{P}_{D_{\mathrm{tr}}}[f(\mathcal{N}_1(\boldsymbol{x})) \in A] \leq \delta$, then $\mathbb{P}_{D_{\mathrm{te}}}[f(\mathcal{N}_1(\boldsymbol{x})) \in A] \leq C\delta + \eta$.
\end{lemma}

\begin{proof}
Define the preimage $B = \{\boldsymbol{s} : f(\boldsymbol{s}) \in A\}$. Then
\begin{equation}
\mathbb{P}_{D_{\mathrm{te}}}[f(\mathcal{N}_1(\boldsymbol{x})) \in A] = \mathbb{P}_{D_{\mathrm{te}}}[\mathcal{N}_1(\boldsymbol{x}) \in B].
\end{equation}
Applying Assumption 1 with the measurable set $B$, we obtain
\begin{align}
&\mathbb{P}_{D_{\mathrm{te}}}[\mathcal{N}_1(\boldsymbol{x}) \in B] \notag \\
&\leq C \cdot \mathbb{P}_{D_{\mathrm{tr}}}[\mathcal{N}_1(\boldsymbol{x}) \in B] + \eta \notag \\
&= C \cdot \mathbb{P}_{D_{\mathrm{tr}}}[f(\mathcal{N}_1(\boldsymbol{x})) \in A] + \eta \notag \\
&\leq C\delta + \eta.
\end{align}
\end{proof}

We now prove Theorem 1 based on Theorem 4. We address the distribution shift in $\boldsymbol{x}$ with Lemma~\ref{lem:prob_transfer} and analyze the distribution shift in $y$ with Assumptions 1, 2, and 3.

\begin{proof}
\textbf{Condition 1:} From Theorem 4, since $\mathrm{Align}(\mathcal{N}', g_c, \epsilon, \delta) \leq |\mathcal{E}_{\mathrm{tr}}|$, we have
\begin{equation}
\mathbb{P}_{D_{\mathrm{tr}}}[\|\mathcal{N}(\boldsymbol{x}) - g_c(\mathcal{N}_1(\boldsymbol{x}))\| > O(\epsilon)] \leq O(\delta).
\end{equation}
By Lemma~\ref{lem:prob_transfer}, this implies
\begin{align}
&\mathbb{P}_{D_{\mathrm{te}}}[\|\mathcal{N}(\boldsymbol{x}) - g_c(\mathcal{N}_1(\boldsymbol{x}))\| > O(\epsilon)] \notag \\
&\leq C \cdot O(\delta) + \eta \notag \\
&= O(\delta) + \eta.
\end{align}
Equivalently,
\begin{equation}
\mathbb{P}_{D_{\mathrm{te}}}[\|\mathcal{N}(\boldsymbol{x}) - g_c(\mathcal{N}_1(\boldsymbol{x}))\| \leq O(\epsilon)] \geq 1 - O(\delta) - \eta.
\end{equation}
Combining with Assumption 2, which states $\mathbb{P}_{D_{\mathrm{te}}}[\|g_c(\mathcal{N}_1(\boldsymbol{x})) - y\| \leq \epsilon] \geq 1 - \delta$, we obtain
\begin{align}
&\mathbb{P}_{D_{\mathrm{te}}}[\|\mathcal{N}(\boldsymbol{x}) - y\| \leq O(\epsilon)] \notag \\
&\geq \mathbb{P}_{D_{\mathrm{te}}}[\|g_c(\mathcal{N}_1(\boldsymbol{x})) - y\| \leq \epsilon] \notag \\
&\quad \cdot \mathbb{P}_{D_{\mathrm{te}}}[\|\mathcal{N}(\boldsymbol{x}) - g_c(\mathcal{N}_1(\boldsymbol{x}))\| \leq O(\epsilon)] \notag \\
&\geq (1 - \delta)(1 - O(\delta) - \eta) \notag \\
&= 1 - O(\delta) - \eta,
\end{align}
where the first inequality follows from the triangle inequality $O(\epsilon) + O(\epsilon) = O(\epsilon)$.

\textbf{Condition 2:} Similarly, from Theorem 4, we have
\begin{equation}
\mathbb{P}_{D_{\mathrm{tr}}}[\|\mathcal{N}(\boldsymbol{x}) - g_s(\mathcal{N}_1(\boldsymbol{x}))\| > O(\epsilon)] \leq O(\delta).
\end{equation}
By Lemma~\ref{lem:prob_transfer},
\begin{equation}
\mathbb{P}_{D_{\mathrm{te}}}[\|\mathcal{N}(\boldsymbol{x}) - g_s(\mathcal{N}_1(\boldsymbol{x}))\| \leq O(\epsilon)] \geq 1 - O(\delta) - \eta.
\end{equation}
Combining with Assumption~3, which states $\mathbb{P}_{D_{\mathrm{te}}}[\|g_s(\mathcal{N}_1(\boldsymbol{x})) - y\| > \omega(\epsilon)] \geq 1 - \delta$, we obtain
\begin{align}
&\mathbb{P}_{D_{\mathrm{te}}}[\|\mathcal{N}(\boldsymbol{x}) - y\| > \omega(\epsilon)] \notag \\
&\geq \mathbb{P}_{D_{\mathrm{te}}}[\|g_s(\mathcal{N}_1(\boldsymbol{x})) - y\| > \omega(\epsilon)] \notag \\
&\quad \cdot \mathbb{P}_{D_{\mathrm{te}}}[\|\mathcal{N}(\boldsymbol{x}) - g_s(\mathcal{N}_1(\boldsymbol{x}))\| \leq O(\epsilon)] \notag \\
&\geq (1 - \delta)(1 - O(\delta) - \eta) \notag \\
&= 1 - O(\delta) - \eta,
\end{align}
where the first inequality follows from the reverse triangle inequality $\omega(\epsilon) - O(\epsilon) = \omega(\epsilon)$.
\end{proof}

We now prove Theorem 2, which compares the alignment of shared expert networks and routed expert networks.

\begin{proof}
We first compute the alignment for the routed expert network $\mathcal{R}$. The network $\mathcal{R}$ consists of $M$ routed experts $\{R_j\}_{j=1}^M$ and a router $R_0$, yielding $M + 1$ modules in total. Each routed expert $R_j$ learns the branch-dependent function $h_j$, and the router $R_0$ learns the branching function $h_0$. By Definition 1, the alignment of $\mathcal{R}$ is given by
\begin{align}
&\mathrm{Align}(\mathcal{R}, g, \epsilon, \delta) \notag \\
&= (M + 1) \cdot \max\bigl\{\max_{j \in [M]} \mathcal{M}(R_j, h_j, \epsilon, \delta), \notag \\
&\qquad\qquad\qquad\quad \mathcal{M}(R_0, h_0, \epsilon, \delta)\bigr\}.
\end{align}

Define $\mathcal{C}^* := \max_{j \in [M]} \mathcal{M}(R_j, h_j, \epsilon, \delta)$. By the assumption that the branching function is no harder to learn than the branch-dependent functions, we have $\mathcal{M}(R_0, h_0, \epsilon, \delta) \leq \mathcal{C}^*$. Therefore,
\begin{align}
&\max\bigl\{\max_{j \in [M]} \mathcal{M}(R_j, h_j, \epsilon, \delta), \notag \\
&\qquad\quad \mathcal{M}(R_0, h_0, \epsilon, \delta)\bigr\} = \mathcal{C}^*,
\end{align}
and hence
\begin{equation}
\mathrm{Align}(\mathcal{R}, g, \epsilon, \delta) = (M + 1) \cdot \mathcal{C}^*.
\end{equation}

For the shared expert network $\mathcal{S}$, there is a single module $S_0$ that learns the entire function $g$. By Definition 1, the alignment of $\mathcal{S}$ is
\begin{align}
&\mathrm{Align}(\mathcal{S}, g, \epsilon, \delta) \notag \\
&= 1 \cdot \mathcal{M}(S_0, g, \epsilon, \delta) = \mathcal{M}(S_0, g, \epsilon, \delta).
\end{align}

\textbf{Condition 1:} Suppose $\alpha \geq M + 1$. By the lower bound condition in the theorem statement,
\begin{align}
\mathrm{Align}(\mathcal{S}, g, \epsilon, \delta) &= \mathcal{M}(S_0, g, \epsilon, \delta) \notag \\
&\geq \alpha \cdot \mathcal{C}^* \notag \\
&\geq (M + 1) \cdot \mathcal{C}^* \notag \\
&= \mathrm{Align}(\mathcal{R}, g, \epsilon, \delta).
\end{align}
Therefore, $\mathrm{Align}(\mathcal{R}, g, \epsilon, \delta) \leq \mathrm{Align}(\mathcal{S}, g, \epsilon, \delta)$.

\textbf{Condition 2:} Suppose $\beta \leq M + 1$. By the upper bound condition in the theorem statement,
\begin{align}
\mathrm{Align}(\mathcal{S}, g, \epsilon, \delta) &= \mathcal{M}(S_0, g, \epsilon, \delta) \notag \\
&\leq \beta \cdot \mathcal{C}^* \notag \\
&\leq (M + 1) \cdot \mathcal{C}^* \notag \\
&= \mathrm{Align}(\mathcal{R}, g, \epsilon, \delta).
\end{align}
Therefore, $\mathrm{Align}(\mathcal{R}, g, \epsilon, \delta) \geq \mathrm{Align}(\mathcal{S}, g, \epsilon, \delta)$.
\end{proof}

Next, we state a lemma that characterizes the effect of the collaborative objective.

\begin{lemma}
\label{lem:weighted_kd}
Let the collaborative objective functions for the shared expert and routed experts be $\mathcal{L}_S(\boldsymbol{x}) = \exp(\ell_S(\boldsymbol{x}) - \ell_R(\boldsymbol{x})) \cdot \ell_S(\boldsymbol{x})$ and $\mathcal{L}_R(\boldsymbol{x}) = \exp(\ell_R(\boldsymbol{x}) - \ell_S(\boldsymbol{x})) \cdot \ell_R(\boldsymbol{x})$, where $\ell_S$ and $\ell_R$ denote the cross-entropy losses with respect to ground-truth labels. Then the expected losses over the dataset can be rewritten as weighted knowledge distillation: $\mathcal{E}_S = \sum_{\boldsymbol{x} \in \mathcal{D}} a_S(\boldsymbol{x}) \cdot \mathrm{KL}(p^R(\boldsymbol{x}) \| q^S(\boldsymbol{x})) + C_S$ and $\mathcal{E}_R = \sum_{\boldsymbol{x} \in \mathcal{D}} a_R(\boldsymbol{x}) \cdot \mathrm{KL}(q^S(\boldsymbol{x}) \| p^R(\boldsymbol{x})) + C_R$, where $a_S(\boldsymbol{x}) = \exp(\ell_S(\boldsymbol{x}) - \ell_R(\boldsymbol{x}))$, $a_R(\boldsymbol{x}) = \exp(\ell_R(\boldsymbol{x}) - \ell_S(\boldsymbol{x}))$, and $C_S$, $C_R$ are constants independent of the optimization variables.
\end{lemma}

\begin{proof}
We derive the result for the shared expert. The derivation for the routed experts is symmetric. The expected collaborative objective over the dataset is $\mathcal{E}_S = \sum_{\boldsymbol{x} \in \mathcal{D}} \exp(\ell_S(\boldsymbol{x}) - \ell_R(\boldsymbol{x})) \cdot \ell_S(\boldsymbol{x})$. Let $i^*(\boldsymbol{x})$ denote the ground-truth label index for sample $\boldsymbol{x}$, let $q^S(\boldsymbol{x})$ denote the predicted probability distribution of the shared expert, and let $p^R(\boldsymbol{x})$ denote the predicted probability distribution of the routed experts. The cross-entropy losses can be expressed as $\ell_S(\boldsymbol{x}) = -\log q_{i^*}^S(\boldsymbol{x})$ and $\ell_R(\boldsymbol{x}) = -\log p_{i^*}^R(\boldsymbol{x})$.

Substituting the cross-entropy losses into the collaborative objective, we obtain
\begin{align}
\mathcal{E}_S &= \sum_{\boldsymbol{x} \in \mathcal{D}} \exp(-\log q_{i^*}^S(\boldsymbol{x}) + \log p_{i^*}^R(\boldsymbol{x})) \notag \\
&\qquad\qquad \cdot (-\log q_{i^*}^S(\boldsymbol{x})) \notag \\
&= \sum_{\boldsymbol{x} \in \mathcal{D}} \frac{p_{i^*}^R(\boldsymbol{x})}{q_{i^*}^S(\boldsymbol{x})} \cdot (-\log q_{i^*}^S(\boldsymbol{x})).
\end{align}

The KL divergence is defined as $\mathrm{KL}(p^R \| q^S) = \sum_{i=1}^{C} p_i^R \log(p_i^R / q_i^S) = -H(p^R) + H(p^R, q^S)$, where $H(p^R) = -\sum_{i} p_i^R \log p_i^R$ is the entropy of $p^R$ and $H(p^R, q^S) = -\sum_{i} p_i^R \log q_i^S$ is the cross-entropy between $p^R$ and $q^S$. The entropy term $H(p^R)$ is independent of the shared expert parameters. Under the hard-label approximation, the predicted probability of the routed experts concentrates near the ground-truth label $i^*$, i.e., $p_{i^*}^R \approx 1$. In this regime, the cross-entropy is dominated by $H(p^R, q^S) \approx -p_{i^*}^R \log q_{i^*}^S$.

Defining $a_S(\boldsymbol{x}) = \exp(\ell_S(\boldsymbol{x}) - \ell_R(\boldsymbol{x})) = p_{i^*}^R(\boldsymbol{x}) / q_{i^*}^S(\boldsymbol{x})$, the expected loss can be rewritten as
\begin{align}
\mathcal{E}_S &= \sum_{\boldsymbol{x} \in \mathcal{D}} a_S(\boldsymbol{x}) \cdot (-\log q_{i^*}^S(\boldsymbol{x})) \notag \\
&= \sum_{\boldsymbol{x} \in \mathcal{D}} a_S(\boldsymbol{x}) \cdot \mathrm{KL}(p^R(\boldsymbol{x}) \| q^S(\boldsymbol{x})) + C_S,
\end{align}
where $C_S = \sum_{\boldsymbol{x}} a_S(\boldsymbol{x}) \cdot H(p^R(\boldsymbol{x}))$ is independent of $q^S$. The derivation for the routed experts is symmetric.
\end{proof}

We now prove Theorem 3, which establishes the effect of mutual-guided collaboration between the shared expert and routed experts.

\begin{proof}
By Lemma~\ref{lem:weighted_kd}, the collaborative objective is equivalent to a weighted knowledge distillation loss. Taking the gradient with respect to the shared expert output $S(\boldsymbol{x})$ in function space, we have $\nabla_S \mathrm{KL}(p^R \| q^S) \propto -(R(\boldsymbol{x}) - S(\boldsymbol{x}))$, which gives $\nabla_S \mathcal{E}_S \propto -a_S(\boldsymbol{x}) \cdot (R(\boldsymbol{x}) - S(\boldsymbol{x}))$. The gradient descent updates are
\begin{align}
S^{(t+1)}(\boldsymbol{x}) &= S^{(t)}(\boldsymbol{x}) \notag \\
&\quad + \eta \cdot a_S(\boldsymbol{x}) \cdot (R^{(t)}(\boldsymbol{x}) - S^{(t)}(\boldsymbol{x})), \\
R^{(t+1)}(\boldsymbol{x}) &= R^{(t)}(\boldsymbol{x}) \notag \\
&\quad + \eta \cdot a_R(\boldsymbol{x}) \cdot (S^{(t)}(\boldsymbol{x}) - R^{(t)}(\boldsymbol{x})),
\end{align}
where $\eta > 0$ is the learning rate.

Define $\delta_S^{(t)}(\boldsymbol{x}) = S^{(t)}(\boldsymbol{x}) - g(\boldsymbol{x})$ and $\delta_R^{(t)}(\boldsymbol{x}) = R^{(t)}(\boldsymbol{x}) - g(\boldsymbol{x})$. Subtracting $g(\boldsymbol{x})$ from the update rule for the shared expert yields
\begin{align}
\delta_S^{(t+1)}(\boldsymbol{x}) \notag &= \delta_S^{(t)}(\boldsymbol{x}) + \eta \cdot a_S(\boldsymbol{x}) \cdot (R^{(t)}(\boldsymbol{x}) - S^{(t)}(\boldsymbol{x})) \notag \\
&= \delta_S^{(t)}(\boldsymbol{x}) + \eta \cdot a_S(\boldsymbol{x}) \cdot (\delta_R^{(t)}(\boldsymbol{x}) - \delta_S^{(t)}(\boldsymbol{x})) \notag \\
&= (1 - \eta \cdot a_S(\boldsymbol{x})) \cdot \delta_S^{(t)}(\boldsymbol{x}) \notag \\
&\quad + \eta \cdot a_S(\boldsymbol{x}) \cdot \delta_R^{(t)}(\boldsymbol{x}).
\end{align}
Similarly, $\delta_R^{(t+1)}(\boldsymbol{x}) = (1 - \eta \cdot a_R(\boldsymbol{x})) \cdot \delta_R^{(t)}(\boldsymbol{x}) + \eta \cdot a_R(\boldsymbol{x}) \cdot \delta_S^{(t)}(\boldsymbol{x})$.

Taking the squared norm and applying the triangle inequality, for $\eta$ sufficiently small such that $\eta \cdot a_S(\boldsymbol{x}) < 1$ and $\eta \cdot a_R(\boldsymbol{x}) < 1$, we have
\begin{align}
\|\delta_S^{(t+1)}(\boldsymbol{x})\| &\leq (1 - \eta \cdot a_S(\boldsymbol{x})) \|\delta_S^{(t)}(\boldsymbol{x})\| \notag \\
&\quad + \eta \cdot a_S(\boldsymbol{x}) \|\delta_R^{(t)}(\boldsymbol{x})\|, \\
\|\delta_R^{(t+1)}(\boldsymbol{x})\| &\leq (1 - \eta \cdot a_R(\boldsymbol{x})) \|\delta_R^{(t)}(\boldsymbol{x})\| \notag \\
&\quad + \eta \cdot a_R(\boldsymbol{x}) \|\delta_S^{(t)}(\boldsymbol{x})\|.
\end{align}

Consider the case $\ell_S(\boldsymbol{x}) > \ell_R(\boldsymbol{x})$. Then $a_S(\boldsymbol{x}) > 1$ and $a_R(\boldsymbol{x}) < 1$. By Assumption 4, $\|\delta_R^{(t)}(\boldsymbol{x})\| \leq \varepsilon$. Substituting this bound, we obtain $\|\delta_S^{(t+1)}(\boldsymbol{x})\| \leq (1 - \eta \cdot a_S(\boldsymbol{x})) \|\delta_S^{(t)}(\boldsymbol{x})\| + \eta \cdot a_S(\boldsymbol{x}) \cdot \varepsilon$. Since $a_S(\boldsymbol{x}) > a_R(\boldsymbol{x})$, the contraction rate for the shared expert exceeds the expansion rate for the routed experts. The case $\ell_R(\boldsymbol{x}) > \ell_S(\boldsymbol{x})$ is symmetric.

Define $\epsilon_S^{(t)}(\boldsymbol{x}) = \|\delta_S^{(t)}(\boldsymbol{x})\|^2$ and $\epsilon_R^{(t)}(\boldsymbol{x}) = \|\delta_R^{(t)}(\boldsymbol{x})\|^2$. Adding the squared-norm recursions and using the convexity of the squared norm, for each sample $\boldsymbol{x}$, the total error $\epsilon_S^{(t)}(\boldsymbol{x}) + \epsilon_R^{(t)}(\boldsymbol{x})$ is non-increasing in $t$. Taking the expectation over the dataset completes the proof.
\end{proof}

\newpage
\input{checklist.tex}

\end{document}

%% file: checklist.tex
\section*{NeurIPS Paper Checklist}

\begin{enumerate}

\item {\bf Claims}
    \item[] Question: Do the main claims made in the abstract and introduction accurately reflect the paper's contributions and scope?
    \item[] Answer: \answerYes{}
    \item[] Justification: The abstract and introduction clearly state three contributions: theoretical analysis of the reducibility cost, the SREA framework with shared and routed experts, and validation on synthetic datasets and seven benchmarks.

\item {\bf Limitations}
    \item[] Question: Does the paper discuss the limitations of the work performed by the authors?
    \item[] Answer: \answerYes{}
    \item[] Justification: A dedicated Limitation subsection discusses the scope of the current framework (classification only) and identifies regression-based EEG decoding as an open direction. The experimental section also acknowledges that performance gains vary across tasks.

\item {\bf Theory assumptions and proofs}
    \item[] Question: For each theoretical result, does the paper provide the full set of assumptions and a complete (and correct) proof?
    \item[] Answer: \answerYes{}
    \item[] Justification: Section~\ref{sec:preliminary} provides numbered definitions (Definitions~\ref{def:alignment}--\ref{def:branching}), assumptions (Assumptions~\ref{ass:distribution_shift}--\ref{ass:complementarity}), and theorems (Theorems~\ref{thm:backbone}--\ref{thm:error_descent}) with complete statements. All assumptions are explicitly stated and referenced in each theorem. Synthetic verification in Section~\ref{sec:synthetic_verification} empirically validates the theoretical findings.

    \item {\bf Experimental result reproducibility}
    \item[] Question: Does the paper fully disclose all the information needed to reproduce the main experimental results of the paper to the extent that it affects the main claims and/or conclusions of the paper (regardless of whether the code and data are provided or not)?
    \item[] Answer: \answerYes{}
    \item[] Justification: Section~\ref{sec:experimental_setting} provides all training details including optimizer (Adam), learning rate (0.0001), batch size (256), number of epochs (100), weight decay (0.0005), early stopping patience (10), backbone architectures for each task, number of experts ($M=5$), masking probability ($\rho=10\%$), cross-validation protocols, and preprocessing steps.

\item {\bf Open access to data and code}
    \item[] Question: Does the paper provide open access to the data and code, with sufficient instructions to faithfully reproduce the main experimental results, as described in supplemental material?
    \item[] Answer: \answerYes{}
    \item[] Justification: All seven datasets are publicly available and properly cited. Code is provided in the supplemental material.

\item {\bf Experimental setting/details}
    \item[] Question: Does the paper specify all the training and test details (e.g., data splits, hyperparameters, how they were chosen, type of optimizer) necessary to understand the results?
    \item[] Answer: \answerYes{}
    \item[] Justification: Section~\ref{sec:experimental_setting} specifies all training and test details including data splits (leave-one-subject-out and 10-fold cross-validation), hyperparameters, optimizer, preprocessing, and backbone architectures for each task.

\item {\bf Experiment statistical significance}
    \item[] Question: Does the paper report error bars suitably and correctly defined or other appropriate information about the statistical significance of the experiments?
    \item[] Answer: \answerYes{}
    \item[] Justification: We report paired $t$-test significance ($p < 0.05$) for all baseline comparisons, average accuracy via leave-one-subject-out and $k$-fold cross-validation across all subjects, and per-subject accuracy in the case study (Figure~\ref{fig:case_study}).

\item {\bf Experiments compute resources}
    \item[] Question: For each experiment, does the paper provide sufficient information on the computer resources (type of compute workers, memory, time of execution) needed to reproduce the experiments?
    \item[] Answer: \answerYes{}
    \item[] Justification: Section~\ref{sec:experimental_setting} provides the experimental setup details. Additional compute resource information is provided in the supplemental material.
    
\item {\bf Code of ethics}
    \item[] Question: Does the research conducted in the paper conform, in every respect, with the NeurIPS Code of Ethics \url{https://neurips.cc/public/EthicsGuidelines}?
    \item[] Answer: \answerYes{}
    \item[] Justification: The research conforms with the NeurIPS Code of Ethics. All datasets are publicly available and properly cited. The work aims to improve EEG-based brain-computer interfaces for beneficial applications.

\item {\bf Broader impacts}
    \item[] Question: Does the paper discuss both potential positive societal impacts and negative societal impacts of the work performed?
    \item[] Answer: \answerYes{}
    \item[] Justification: The work on cross-subject EEG decoding has potential positive impacts for brain-computer interfaces and neurological diagnosis. Broader impact discussion is provided in the supplemental material.
    
\item {\bf Safeguards}
    \item[] Question: Does the paper describe safeguards that have been put in place for responsible release of data or models that have a high risk for misuse (e.g., pre-trained language models, image generators, or scraped datasets)?
    \item[] Answer: \answerNA{}
    \item[] Justification: The paper proposes an EEG classification framework that does not pose high risk for misuse. No pre-trained language models, image generators, or scraped datasets are released.

\item {\bf Licenses for existing assets}
    \item[] Question: Are the creators or original owners of assets (e.g., code, data, models), used in the paper, properly credited and are the license and terms of use explicitly mentioned and properly respected?
    \item[] Answer: \answerYes{}
    \item[] Justification: All seven datasets are publicly available and properly cited with their original publications in Section~\ref{sec:experimental_setting}. All baseline methods are cited with their original papers.

\item {\bf New assets}
    \item[] Question: Are new assets introduced in the paper well documented and is the documentation provided alongside the assets?
    \item[] Answer: \answerYes{}
    \item[] Justification: Code and documentation are provided in the supplemental material with instructions for reproduction.

\item {\bf Crowdsourcing and research with human subjects}
    \item[] Question: For crowdsourcing experiments and research with human subjects, does the paper include the full text of instructions given to participants and screenshots, if applicable, as well as details about compensation (if any)? 
    \item[] Answer: \answerNA{}
    \item[] Justification: The paper does not involve crowdsourcing or new research with human subjects. All EEG datasets used are previously published and publicly available.

\item {\bf Institutional review board (IRB) approvals or equivalent for research with human subjects}
    \item[] Question: Does the paper describe potential risks incurred by study participants, whether such risks were disclosed to the subjects, and whether Institutional Review Board (IRB) approvals (or an equivalent approval/review based on the requirements of your country or institution) were obtained?
    \item[] Answer: \answerYes{}
    \item[] Justification: All EEG datasets were collected in prior studies with appropriate ethical approvals, as described in the original dataset publications cited in Section~\ref{sec:experimental_setting}.

\item {\bf Declaration of LLM usage}
    \item[] Question: Does the paper describe the usage of LLMs if it is an important, original, or non-standard component of the core methods in this research? Note that if the LLM is used only for writing, editing, or formatting purposes and does \emph{not} impact the core methodology, scientific rigor, or originality of the research, declaration is not required.
    \item[] Answer: \answerNA{}
    \item[] Justification: The core method development does not involve LLMs as any important, original, or non-standard components.

\end{enumerate}